\documentclass[paper, 10 pt, conference]{ieeeconf}%[twocolumn]{IEEEtran}
\usepackage[T1]{fontenc}% optional T1 font encoding

\IEEEoverridecommandlockouts                          
\usepackage{amsmath} % assumes amsmath package installed
\usepackage{amssymb}  % assumes amsmath package installed
\usepackage{xcolor}
\usepackage{caption}
\usepackage{graphicx}
\usepackage{mathtools}
\usepackage{caption}
\usepackage{floatrow}
\usepackage{tikz}

\usepackage{algorithmicx}
    \usepackage[ruled]{algorithm}
    \usepackage{algpseudocode}

\usepackage{pgfplots}
\usepackage{xcolor}
\usetikzlibrary{matrix,arrows,calc,positioning,shapes,decorations.pathreplacing}
\usepackage{graphicx}

\usepackage{amsmath} % assumes amsmath package installed

\newcommand{\R}{\mathbb{R}}

\DeclareMathOperator{\spn}{span}

\DeclareMathOperator{\rank}{rank}

\usepackage{enumerate}
\usepackage[all,tips]{xy}
\SelectTips{cm}{11}

\newtheorem{theorem}{Theorem}

\newtheorem{definition}{Definition}
\newtheorem{remark}{Remark}

\newtheorem{lemma}{Lemma}

\begin{document}

\title{
Geometric Stabilization of Virtual Nonlinear Nonholonomic Constraints}

\author{Efstratios Stratoglou, Alexandre Anahory Simoes,  Anthony Bloch, and Leonardo Colombo 
%\thanks{$^{*}$ Authors contributed equally.} 
\thanks{E. Stratoglou (ef.stratoglou@alumnos.upm.es) is with Universidad Polit\'ecnica de Madrid (UPM), 28006 Madrid, Spain.} 
\thanks{A. Anahory Simoes  (alexandre.anahory@ie.edu) is with the School of Science and Technology, IE University, Spain.}
\thanks{A. Bloch (abloch@umich.edu) is with Department of Mathematics, University of Michigan, Ann Arbor, MI 48109, USA.}
\thanks{L. Colombo (leonardo.colombo@car.upm-csic.es) is with Centre for Automation and Robotics (CSIC-UPM), Ctra. M300 Campo Real, Km 0,200, Arganda
del Rey - 28500 Madrid, Spain.}
\thanks{The authors acknowledge financial support from Grant PID2022-137909NB-C21 funded by MCIN/AEI/ 10.13039/501100011033. A.B. was partially supported by NSF grant  DMS-2103026, and AFOSR grants FA
9550-22-1-0215 and FA 9550-23-1-0400.}}%

\maketitle

\begin{abstract}
    In this paper, we address the problem of stabilizing a system around a desired manifold determined by virtual nonlinear nonholonomic constraints. Virtual constraints are relationships imposed on a control system that are rendered invariant through feedback control. Virtual nonholonomic constraints represent a specific class of virtual constraints that depend on the system’s velocities in addition to its configurations. We derive a control law under which a mechanical control system achieves exponential convergence to the virtual constraint submanifold, and rendering it control-invariant. The proposed controller’s performance is validated through simulation results in two distinct applications: flocking motion in multi-agent systems and the control of an unmanned surface vehicle (USV) navigating a stream.

\end{abstract}

%\begin{IEEEkeywords} Geometric control, Virtual constraints, Affine connection control systems, Nonholonomic  systems. \end{IEEEkeywords}

\IEEEpeerreviewmaketitle

%\textcolor{magenta}{Tony: just a thought about the phrasing of the problem: we are not really stabilizing the constraints themselves but the 
%submanifold (subset?) of the phase space defined by the constraints. It is a bit clumsy to say this so I understand why one might not want to keep 
%saying it  but perhaps one should at least the first couple of times?}
\section{Introduction}
Virtual constraints are constraints enforceable through the application of external forces. Over time, virtual holonomic constraints have proven to be a powerful tool for motion control, particularly in the context of bipedal robots \cite{westervelt2018feedback}. In \cite{CanudasdeWit:ontheconcept:2004}, these constraints were applied to design orbitally stable feedback control laws for tasks such as balancing and walking. In subsequent years, this methodology was expanded to support motion planning in more general robotic systems, as discussed in \cite{Consol:Constal:2015} and the references therein.

Virtual nonholonomic constraints represent a specific type of virtual constraints that are defined based on the velocities of the system rather than solely on its configurations. These constraints were first introduced in \cite{griffin2015nonholonomic} as a method for designing velocity-dependent swing foot placements in bipedal robots (see also \cite{hamed2019nonholonomic}, \cite{horn2018hybrid}, \cite{moran2021energy} for further applications).

In \cite{VNNHC}, we introduced the concept of virtual nonlinear nonholonomic constraints within a geometric framework. Specifically, we defined a controlled invariant manifold associated with an affine connection mechanical control system. We established the existence and uniqueness of a control law that enforces a virtual nonlinear nonholonomic constraint and characterized the trajectories of the resulting closed-loop system as solutions of a mechanical system governed by an induced constrained connection. Furthermore, we identified the conditions under which nonholonomic dynamics can emerge from virtual nonholonomic constraints. 

In this paper, we focus on the stabilization of systems around desired manifolds of the phase space, determined by virtual nonlinear nonholonomic constraints. We prove the existence of a control law ensuring that the system adheres to the constraints. Additionally, we show that if the system already satisfies the constraints at a given point, the control law aligns with the unique control law derived in \cite{VNNHC}, which guarantees the existence of a virtual nonlinear nonholonomic constraint.

In our previous paper \cite{stabilization}, we studied the stabilization problem for the case of linear nonholonomic constraints. In this paper, we extend the results to nonlinear constraints and we derive more flexible control laws containing gain matrices able to accommodate different rates of convergence. In particular, these advances allow us to cover more interesting cases in engineering and could be used to tackle more challenging problems in robotic locomotion. The main difficulty with the results achieved in this paper lies on the complex interaction of nonlinear constraints with the dynamics.

The remainder of the paper is organized as follows. In Section \ref{sec2} we introduce virtual nonlinear nonholonomic constraints in an affine connection (coordinate-free) framework. In Section \ref{sec3} we address the problem of stabilizing a system around desired virtual nonlinear nonholonomic constraints. In Section \ref{sec4} the proposed controller’s performance is validated through simulation results in two distinct applications: flocking motion in multi-agent systems and the control of an unmanned surface vehicle (USV) navigating a stream. We conclude the paper in Section \ref{conc} with some directions for the future. 

\section{Virtual Nonholonomic constraints}\label{sec2}
\subsection{Mechanical systems in a coordinate-free setting}

We begin with a Lagrangian system on an $n$-dimensional configuration space $Q$ and Lagrangian $L:TQ\to\mathbb{R}$. We assume the Lagrangian has the mechanical form $L=K-V\circ\pi$, where $K$ is a function on $TQ$, the tangent bundle of $Q$, describing the kinetic energy of the system, that is, $K=\frac{1}{2}\mathcal{G}(v_q, v_q)$, where $\mathcal{G}$ is the Riemannian metric on $Q$, $V:Q\to\mathbb{R}$ is a function on $Q$ representing the potential energy, and $\pi:TQ\to Q$ is the tangent bundle projection, locally given by $\pi(q,\dot{q})=q$ with $(q,\dot{q})$ denoting local coordinates on $TQ$. In addition, we denote by $\mathfrak{X}(Q)$ the set of vector fields on $Q$ and by $\Omega^{1}(Q)$ the set of $1$-forms on $Q$. If $X, Y\in\mathfrak{X}(Q),$ then
$[X,Y]$ denotes the Lie bracket of vector fields.

In the case that $Q$ is a Riemannian manifold, there is a unique connection $\nabla^{\mathcal{G}}:\mathfrak{X}(Q)\times \mathfrak{X}(Q) \rightarrow \mathfrak{X}(Q)$ called the \textit{Levi-Civita connection} satisfying 
\begin{enumerate}
\item $[ X,Y]=\nabla_{X}^{\mathcal{G}}Y-\nabla_{Y}^{\mathcal{G}}X$ (symmetry)
\item $X(\mathcal{G}(Y,Z))=\mathcal{G}(\nabla_{X}^{\mathcal{G}}(Y,Z)+\mathcal{G}(Y,\nabla_{X}^{\mathcal{G}}Z)$ (compatibillity of the metric).
\end{enumerate}

The Levi-Civita connection helps us describe the trajectories of a mechanical Lagrangian system. The trajectories $q:I\rightarrow Q$ of a mechanical Lagrangian determined by $L$ satisfy the following equations
\begin{equation}\label{ELeq}
    \nabla_{\dot{q}}^{\mathcal{G}}\dot{q} + \text{grad}_{\mathcal{G}}V(q(t)) = 0,
\end{equation}
where the vector field $\text{grad}_{\mathcal{G}}V\in\mathfrak{X}(Q)$ is characterized by $\mathcal{G}(\text{grad}_{\mathcal{G}}V, X) = dV(X), \; \mbox{ for  every } X \in
\mathfrak{X}(Q)$. If $V$ vanishes, then the trajectories of the mechanical system are just the geodesics with respect to the connection $\nabla^{\mathcal{G}}$.

\subsection{Nonlinear nonholonomic mechanics }

A nonlinear nonholonomic constraint on a mechanical system is a submanifold $\mathcal{M}$ of the tangent bundle $TQ$ from which the velocity of the system can not leave. Mathematically, the constraint may be written as the set of points where a function of the type $\Phi:TQ \rightarrow \mathbb{R}^{m}$ vanishes, where $m < n=\dim Q$. That is, $\mathcal{M}=\Phi^{-1}(\{0\})$. If every point in $\mathcal{M}$ is regular, i.e., the tangent map $T_{p}\Phi$ is surjective for every $p\in \mathcal{M}$, then $\mathcal{M}$ is a submanifold of $TQ$ with dimension $2n-m$ by the regular level set theorem \cite{FoM}. 

Let $\Phi = (\phi^{1}, \dots, \phi^{m})$ denote the coordinate functions of the constraint $\phi$. The coordinate expression of the equations of motion of a system with nonholonomic constraints are called Chetaev's equations and they are given by
\begin{equation}
    \begin{split}
        & \frac{d}{dt}\left(\frac{\partial L}{\partial \dot{q}}\right)-\frac{\partial L}{\partial q}=\lambda_{a} \frac{\partial \phi^{a}}{\partial \dot{q}}, \,\,\, \phi^{a}(q,\dot{q}) = 0,
    \end{split}
    \label{Chetaev:eq}
\end{equation}
(see \cite{MdLeon}, \cite{bloch2015nonholonomic}). The right hand side acts like a constraint force that forces the system to remain inside the constraint submanifold. From the physical perspective, these forces are characterized by their doing no mechanical work on the system. Here $\lambda\in\mathbb{R}^{m}$ denotes the Lagrange multiplier.

In the following, we will consider a mapping that to each point $v_{q}$ on the submanifold $\mathcal{M}$ assigns a vector subspace of $T_{q}Q$. In differential geometry, such a map resembles a distribution on $Q$ restricted to $\mathcal{M}$, but unlike a distribution it also depends on the velocity. Thus we will call it a velocity-dependent distribution \cite{VNNHC}. From now on, let $S(v_{q})$ be a subspace of $T_{q}Q$, with $v_{q}\in \mathcal{M}$, defined by $$S(v_{q})= \{ X \in T_{q}Q \ | \ \Big{\langle} \frac{\partial \phi^{a}}{\partial \dot{q}^{i}}(v_q) dq^{i}, X \Big{\rangle} = 0, \ a = 1,\ldots, m \}.$$  $S(v_{q})$ act as a linearization of the constraint submanifold $\mathcal{M}$ at each point $v_{q}$. Equations \eqref{Chetaev:eq} can be written in Riemannian form using a geodesic-like equation as follows.

\begin{theorem}[\cite{VNNHC}]
    A curve $q:I\rightarrow Q$ is a solution of Chetaev's equations for a mechanical type Lagrangian with kinetic energy determined by a Riemannian metric $\mathcal{G}$ on $Q$ and a potential function $V$ if and only if $\Phi(q,\dot{q})=0$ and it satisfies the equation
   \begin{equation}\label{Chetaev's eqns}
        \nabla_{\dot{q}}\dot{q} + \text{grad } V \in S(\dot{q})^{\bot},
    \end{equation}where $S(\dot{q})^{\bot}$ is the orthogonal velocity-dependent distribution to $S(\dot{q})$ with respect to the Riemannian metric $\mathcal{G}$, and $\nabla$ is the corresponding Levi-Civita connection.
\end{theorem}

\subsection{Virtual nonholonomic constraints}

For a Riemannian metric $\mathcal{G}$ on $Q$, we can use its non-degeneracy property to define the musical isomorphism $\flat_{\mathcal{G}}:\mathfrak{X}(Q)\rightarrow \Omega^{1}(Q)$ defined by $\flat_{\mathcal{G}}(X)(Y)=\mathcal{G}(X,Y)$ for any $X, Y \in \mathfrak{X}(Q)$. Also, denote by $\sharp_{\mathcal{G}}:\Omega^{1}(Q)\rightarrow \mathfrak{X}(Q)$ the inverse musical isomorphism, i.e., $\sharp_{\mathcal{G}}=\flat_{\mathcal{G}}^{-1}$. In local coordinates, $\flat_{\mathcal{G}}(X^{i}\frac{\partial}{\partial q^{i}})=\mathcal{G}_{ij}X^{i}dq^{j}$ and $\sharp_{\mathcal{G}}(\alpha_{i}dq^{i})=\mathcal{G}^{ij}\alpha_{i}\frac{\partial}{\partial q^{j}}$, where $\mathcal{G}^{ij}$ is the inverse matrix of $\mathcal{G}_{ij}$.

Given an external force $F^{0}:TQ\rightarrow T^{*}Q$ and a control force $F:TQ\times U \rightarrow T^{*}Q$ of the form
$\displaystyle{
    F(q,\dot{q},u) = \sum_{a=1}^{m} u_{a}f^{a}(q , \dot{q})}$ where $f^{a}(q,\dot{q})\in T^{*}Q$ with $m<n$, $U\subset\mathbb{R}^{m}$ the set of controls and $u_a\in\mathbb{R}$ with $1\leq a\leq m$ the control inputs, consider the associated mechanical control system
\begin{equation}\label{mechanical:control:system}
    \nabla_{\dot{q}}\dot{q} =Y^0(q,\dot{q})+u_{a}Y^{a}(q, \dot{q}),
\end{equation}
where $Y^0(q,\dot{q})=\sharp_{\mathcal{G}} (F^0(q, \dot{q}))$ and $Y^{a}=\sharp_{\mathcal{G}} (f^{a}(q, \dot{q})).$

Equation \eqref{mechanical:control:system} forms a system of second-order differential equations whose solutions are the trajectories of a vector field of the form \cite{VNNHC}
\begin{equation}\label{SODE}\Gamma(q, \dot{q}, u)=G(q,\dot{q})+u_{a}(Y^{a})_{(q,\dot{q})}^{V}.\end{equation}
We call each $Y^{a}=\sharp_{\mathcal{G}}(f^{a})$ a control force vector field. Here $G$ is the vector field determined by the unactuated forced mechanical system
$\nabla_{\dot{q}}\dot{q} =Y^0(q,\dot{q})$.
    %\nabla^{\mathcal{G}}_{\dot{q}(t)} \dot{q}(t) =Y^{0}(q(t),\dot{q}(t))

\begin{definition}
    The distribution $\mathcal{F}\subseteq TQ$ generated by the vector fields  $Y^{a}=\sharp_{\mathcal{G}}(f^{a})$ is called the \textit{input distribution} associated with the mechanical control system \eqref{mechanical:control:system}.
\end{definition}

%Now we will define the concept of virtual nonholonomic constraint.

\begin{definition}
A \textit{virtual nonholonomic constraint} associated with the mechanical control system \eqref{mechanical:control:system} is a controlled invariant submanifold $\mathcal{M}\subseteq TQ$ for that system, that is, %$\mathcal{M}\subseteq TQ$ is said to be controlled invariant for the controlled system \eqref{lagrangian:control:system} if
there exists a control function $\hat{u}:\mathcal{M}\rightarrow \mathbb{R}^{m}$ such that the solution of the closed-loop system satisfies $\psi_{t}(\mathcal{M})\subseteq \mathcal{M}$, where $\psi_{t}:TQ\rightarrow TQ$ denotes its flow.
\end{definition}

%Now we present the theorem for existence and uniqueness of a control law making the constraint submanifold $\mathcal{M}$ is control invariant given in ....

\begin{theorem}[\cite{VNNHC}]\label{main:theorem1}
If the velocity-dependent distribution, $S(v_q)$, is transversal to the control input distribution $\mathcal{F}$ and $T_{v_{q}}\mathcal{M}\cap \mathcal{F}^V=\{0\}$, then there exists a unique smooth control function making $\mathcal{M}$ a virtual nonholonomic constraint associated with the mechanical control system \eqref{mechanical:control:system}.
\end{theorem}

%\textcolor{red}{Include example for linear case: the roller race}
%\color{blue}

\section{Stabilization of Virtual  Nonholonomic Constraints}\label{sec3}

%A nonlinear nonholonomic constraint on a mechanical system is a submanifold $\mathcal{M}$ of the tangent bundle $TQ$ from which the velocity of the system can not leave. Mathematically, the constraint may be written as the set of points where a function of the type $\phi:TQ \rightarrow \mathbb{R}^{m}$ vanishes, where $m < n=\dim Q$. That is, $\mathcal{M}=\phi^{-1}(\{0\})$. If every point in $\mathcal{M}$ is regular, i.e., the tangent map $T_{p}\phi$ is surjective for every $p\in \mathcal{M}$, then $\mathcal{M}$ is a submanifold of $TQ$ with dimension $2n-m$ by the regular level set theorem. 

Let $\Phi = (\phi^{1}, \dots, \phi^{m})$ denote the coordinate functions of the constraint $\phi$. The interplay between the submanifold $\mathcal{M}$ and a distribution $\mathcal{F}$ is given as follows

\begin{lemma}\label{matrix C^ab}
     Let $\mathcal{M}$ be the constraint submanifold defined by $\phi(q,\dot{q})=0$ and $\mathcal{S}(v_q)$ be the velocity-dependent distribution. Suppose that $\mathcal{S}(v_q)$ is transversal to the input distribution $\mathcal{F}$ generated by the vector fields $\{Y^{b}\}$ and that $T_{v_q}\mathcal{M}\cap\mathcal{F}^V=\{0\}$, then the matrix
    $C^{ab}=(Y^{a})^{V}(\phi^{b})$
    where $\phi^b$ are the components of $\phi(q,\dot{q})$, is invertible with smooth inverse.
\end{lemma}

\begin{proof}
It is easily seen that $(Y^{a})^{V}(\phi^{b}) = d\phi^b(Y^a)^V=\frac{\partial\phi^b}{\partial\dot{q}}(Y^a)^V$. Since the columns of the matrix $C^{ab}$ are linearly independent, $C^{ab}$ has full rank, namely, $\displaystyle{c_{1}\begin{bmatrix} \frac{\partial\phi^1}{\partial \dot{q}}(Y^{1}) \\
        \vdots \\
        \frac{\partial\phi^m}{\partial \dot{q}}(Y^{1}) \end{bmatrix} + \cdots + c_{m}\begin{bmatrix} \frac{\partial\phi^1}{\partial \dot{q}}(Y^{m}) \\
        \vdots \\
        \frac{\partial\phi^m}{\partial \dot{q}}(Y^{m}) \end{bmatrix}= 0}$, which is equivalent to
    \begin{equation*}
        \begin{bmatrix} \frac{\partial\phi^1}{\partial \dot{q}}(c_{1}Y^{1}+\cdots + c_{m}Y^{m}) \\
        \vdots \\
        \frac{\partial\phi^m}{\partial \dot{q}}(c_{1}Y^{1}+\cdots + c_{m}Y^{m}) \end{bmatrix}=0.
    \end{equation*}
By transversality we have $T_{v_q}\mathcal{M}\cap \mathcal{F}^V = \{0\}$ which implies that $c_{1}Y^{1}+\cdots + c_{m}Y^{m}=0$. Since $\{Y^{i}\}$ are linearly independent we conclude that $c_{1}=\cdots=c_{m}=0$ and $C^{ab}$ has full rank. But, since $C^{ab}$ is an $m\times m$ matrix, and $\mathcal{M}$ is a constrained submanifold, it must be invertible
\end{proof}

%\color{black}

%$\left(\begin{tabular}{l}
%Note here that \\
%$\phi^b(q,\dot{q})=\mu^b_i\dot{q}^i + k_i(q) \quad\text{ are the constraints, }$ \\

%$\mu^{b}(q,\dot{q})=\mu^b_i(q,\dot{q})dq^i \quad\text{ are the one forms defining } \D,$ \\

%$\text{ and } (Y^a)^V=Y^a_i\frac{\partial}{\partial \dot{q}^i}.$\\

%So, $C^{ab}=(Y^{a})^{V}(\phi^{b}) = \mu^{b}(Y^{a}).$
%\end{tabular}\right)$ 

%\vspace{0.5cm}

\begin{theorem}\label{main:theorem}
    Given a mechanical control system of the form \eqref{mechanical:control:system} and a virtual constraint submanifold $\mathcal{M}$ determined by $\phi^b=0, \; b=1,\dots, m$, suppose that the input distribution $\mathcal{F}$, generated by the vector field $\{Y^{a}\}$, is transversal to velocity-dependent distribution, $S(v_q)$ . 
    
    Then, $u^{*}:TQ \to \R^{m}$ given by the expression
    \begin{equation}\label{stab:contr:law}
        u_{a}^{*} = C_{a b}(-K\phi^{b}-G(\phi^{b})),
    \end{equation}
    where $G$ is the geodesic vector field, $C_{ab}$ is the inverse matrix of $C^{ab}=(Y^{a})^{V}(\phi^{b})$ and $K$ is a diagonal matrix with positive design parameters $k_i$, $i=1,\dots, m$, satisfies
    \begin{enumerate}
        \item $\phi^{b}\to 0$ exponentially fast along the system trajectories, for $b=1,\ldots, m$.
        \item $u^{*}|_{\mathcal{M}}$ is the unique control law whose existence is guaranteed by Theorem \ref{main:theorem1}.
    \end{enumerate}
\end{theorem}

\begin{proof}
    We show at first that the control law given by \eqref{stab:contr:law} exponentially stabilizes the controlled system \eqref{mechanical:control:system}, in the sense that it drives the system into complying with the constraint exponentially fast. A trajectory, $q$, of the closed loop system \eqref{mechanical:control:system} is an integral curve of the vector field $\Gamma$ of the form \eqref{SODE}. Note that $G$ is the geodesic vector field associated to the Riemannian metric $\mathcal{G}$. Locally, we write
    \[\Gamma(q,\dot{q})=\dot{q}^i\frac{\partial}{\partial q^i}+\left(-\Gamma^j_{ik}\dot{q}^j\dot{q}^k + Y^0_i + u^aY^a_i\right)\frac{\partial}{\partial\dot{q}^i},\] where $\Gamma^j_{ik}$ are the Christoffel symbols. Since the curve $q$ should exponentially vanish the constraint we examine $\phi^b$ along the vector filed $\Gamma.$ Namely,
     \begin{equation*}
            \begin{split}
                \Gamma(\phi^{b}) & = G(\phi^{b}) + u_{a}^{*}(Y^{a})^{V}(\phi^{b}) 
                 = G(\phi^{b}) + C^{ab}u_{a}^{*},
            \end{split}
        \end{equation*}
    where $C^{ab}$ is as in Lemma \ref{matrix C^ab}. Using the control law provided by \eqref{stab:contr:law} we have
    \begin{equation*}
            \begin{split}
                \Gamma(\phi^{b}) & = G(\phi^{b}) + C^{ab}C_{ab}(- K\phi^{b} - G(\phi^{b})) 
                 = - K\phi^{b}.
            \end{split}
        \end{equation*}
    Thus, since $\Gamma(\phi^{b}) = - K\phi^{b}$ for $b=1,\dots,m$, the real valued functions $\phi^{b}$ along the curve $q$ of the closed-loop system \eqref{mechanical:control:system}, $\phi^b(q(t),\dot{q}(t))$ satisfy $\phi^b(t)=\phi^b(0)e^{-Kt}$ for $t\in[a,b].$ Therefore, for every $b=1,\dots, m$, $\phi^b$ converges exponentially fast to zero.
    
    Secondly, we prove that the law given by \eqref{stab:contr:law} is the unique control law guaranteed by Theorem \ref{main:theorem1}.  For any value on the constraint manifold, $v_q\in\mathcal{M}$, the closed-loop system with this control law reads, $\Gamma(v_{q})(\phi^{b}) = - K\phi^{b}(v_{q})$. But, since $v_q\in\mathcal{M}$ the right-hand side of the last equation vanishes. Thus, $\Gamma$ is tangent to $\mathcal{M}$. By uniqueness of the control law given in Theorem $\ref{main:theorem1}$, we have that $u^{*}|_{\mathcal{M}}$ is the unique control law turning $\mathcal{M}$ into a virtual nonholonomic constraint.\end{proof}

\begin{remark}
    When affine constraints are considered for each $q\in Q$ the velocities belong to an affine subspace $\mathcal{A}_q$ of the tangent space $T_qQ.$ Thus, $\mathcal{A}_q$ can be written as a sum of a vector field $X\in\mathfrak{X}(Q)$ and a nonintegrable distribution $\mathcal{D}$ on $Q$, i.e. $\mathcal{A}_q=X(q)+\mathcal{D}_q$, where $\mathcal{D}$ is of constant rank $r$, with $1<r<n$. In this case, we say that the affine space $\mathcal{A}_q$ is modeled on the vector subspace $\mathcal{D}_q$. In local coordinates $\mathcal{D}$ can be expressed as the null space of a $q$-dependent matrix $S(q)$ of dimension $m\times n$ and $\rank S(q)=m,$ with $m=n-r$ as $\mathcal{D}_q=\{\dot{q}\in T_qQ : S(q)\dot{q}=0\}$. The rows of $S(q)$ can be represented by the coordinate functions of $m$ independent 1-forms $\mu^b=\mu^b_idq^i$, $1\leq i\leq n, \; 1\leq b\leq m$.
The affine distribution is $\mathcal{A}_q=\{\dot{q}\in T_qQ : S(q)(\dot{q}-X(q))=0\},$ hence  $\mathcal{A}=\{(q,\dot{q})\in TQ : \Phi(q,\dot{q})=0\}$, with $\Phi(q,\dot{q})=S(q)\dot{q}+Z(q)$ and $Z(q)=-S(q)X(q)\in\R^m$ (see \cite{Sansonetto} for instance and \cite{affine}). Let $Z(q)=[z_1(q), \dots , z_m(q)]^T$ then the components of the affine constraints are given by the real value functions $\phi^b(q,\dot{q})=\mu^b_i\dot{q}^i + z_b(q)$ where $\mu^b_i\dot{q}^i$ are the entries of $S(q)\dot{q}$, $b=1,\dots ,m$. Note that $\phi(q,\dot{q})=0$ is equivalent to $\phi^b=0$ for all $b.$ 
In this case the matrix of Lemma \ref{matrix C^ab} reads $C^{ab}=(Y^a)^V(\phi^b)=\mu^b(Y^a)$ where $\mu^b$ are the 1-forms that define the distribution $\mathcal{D}$.

The case in which we consider linear nonholonomic constraints is a particular case of the aforementioned with $Z(q)=0$ (see \cite{virtual} and \cite{stabilization}).
\end{remark}

\begin{remark}
    If the controlled system \eqref{mechanical:control:system} is subject to holonomic constraints determined by the points $q\in Q$ where $\phi^{b}_h(q)=0$, for a given set of functions $\phi^{b}_h:Q\to\R$ with $b=1,\ldots, m$, instead of nonholonomic ones, then a similar derivation might be used to adapt the control law given in Theorem \ref{main:theorem} resulting in a new control law with the form \[u_{a}^{*} = C_{a b}\left(-K_1\phi^{b}_h - K_2\dot{\phi}^{b}_h-G(\phi^{b}_h)\right),\] to exponentially converge to zero, where $K_i$, $i=1,2$ are diagonal matrices with positive design parameters.
\end{remark}

\color{black}

\section{Applications}\label{sec4}

\subsection{Geometric stabilization for flocking motion}

Flocking, swarming, and schooling are common emergent
collective motion behaviors exhibited in nature. These natural collective behaviors can be leveraged in multirobot systems to safely transport large cohesive groups of robots within a workspace. To capture these effects, Reynolds introduced three heuristic rules: cohesion; alignment; and
separation, to reproduce flocking motions in computer graphics in~\cite{flocking}.

In this application, we study the stabilization problem of virtual nonlinear nonholonomic constraints to impose alignment motion in a multiagent system. Consider four particles moving under the influence of gravity and which we desire to constrain to move with parallel (aligned) velocity. One of the particles will not be controlled and, as a consequence, the remaining particles will try to align their velocities with respect to the uncontrolled particle velocity. Suppose that the motion of the particles evolves in a plane parametrized by $(x,z)$. The position of each particle is given by $q_i=(x_i,z_i)$, with $i=1,2,3,4$, so the configuration space can be considered as $Q=\R^8$ with $q=(q_1,q_2,q_3,q_4)\in Q$.

The Lagrangian $L:\R^{8}\times\R^8\to\R,$ is given by $\displaystyle{L(q,\dot{q})=\sum_{i=1}^{4}\frac{1}{2}m_i\dot{q}_i^2 - m_igz_i}$ where $m_i, i=1,2,3,4$ are the masses of the particles, respectively. The constraint is given by the equation $\Phi:\R^{8}\times\R^8\to\R^3, \; \Phi(q,\dot{q})=[\phi^1 \; \phi^2 \; \phi^3]^T$ where
$\phi^b(q,\dot{q})=\dot{x}_4\dot{z}_b - \dot{x}_b\dot{z}_4$,
$b=1,2,3$ and the control force is just $F:\R^{8}\times\R^8\times\R\to \R^{8}\times\R^8$ given by 
$F(q,\dot{q},u)=u_aF^a$, with $F^a=(q,\dot{q})=dx_a + dz_a$, $a=1,2,3$. The controlled Euler-Lagrange equations are 
\begin{equation}
    \begin{split}
         m_i \ddot{x}_i &= u_i, \,
         m_i\ddot{z}_i +m_ig = u_i, \quad i=1,2,3,\\
         m_4\ddot{x}_4 &= 0, \quad
         m_4 \ddot{z}_2 + m_4g = 0,
    \end{split}
\end{equation}
 The constraint manifold is $\mathcal{M}=\{(q,\dot{q})\in \R^{8}\times\R^8\; :\; \Phi(q,\dot{q})=0\}$ and its tangent space, at every point $(q,\dot{q})\in\mathcal{M}$, is given by        $T_{(q,\dot{q})}\mathcal{M}=\{v\in (\R^8\times\R^8)\times (\R^8\times\R^8)\color{black}\; :\; d\Phi(v)=0\}$.

For the input distribution $\mathcal{F}$ we have $\mathcal{F}=\spn\{Y^a\}$ where $Y^a$, $a=1,2,3$ are the vector fields $\displaystyle{Y=\frac{1}{m_a}\frac{\partial}{\partial x_a} + \frac{1}{m_a}\frac{\partial}{\partial z_a}}.$
Note here that the vertical lift of the input distribution, $\mathcal{F}^V$, which is generated by $\displaystyle{Y^V=\frac{1}{m_a}\frac{\partial}{\partial \dot{x}_a} + \frac{1}{m_a}\frac{\partial}{\partial \dot{z}_a}},$ is transversal to the tangent space of the constraint manifold, $T\mathcal{M}$. Using the control law in Theorem \ref{main:theorem} stabilizing the constraint manifold and making it invariant is
    \begin{equation*}
            \hat{u}=g\left(\dot{x}_4-\dot{z}_4\right)^{-1}diag(m_1,m_2,m_3)\begin{pmatrix}
                \dot{x}_4-\dot{x}_1 \\
                \dot{x}_4-\dot{x}_2 \\
                \dot{x}_4-\dot{x}_3
            \end{pmatrix}.
    \end{equation*}

For the control law that stabilizes the system we use Theorem \ref{main:theorem}, so $u_{a}^{*} = C_{a b}(-\phi^{b}-G(\phi^{b})),$ where \begin{equation*}
    C^{ab}=(Y^a)^V(\phi^b) = (\dot{x}_4 - \dot{z}_4)diag\left(\frac{1}{m_1},\frac{1}{m_2},\frac{1}{m_3}\right)
\end{equation*}

For $i=1,\ldots,4$ the vector field $G$ which is determined by the unactuated forced mechanical system is locally given by
$\displaystyle{G=\dot{x}_i\frac{\partial}{\partial x_i} + \dot{z}_i\frac{\partial}{\partial z_i} -g\frac{\partial}{\partial\dot{z}_i}}$, hence, $G(\hat{\mu}^b)= \begin{pmatrix}
                \dot{x}_4-\dot{x}_1 \\
                \dot{x}_4-\dot{x}_2 \\
                \dot{x}_4-\dot{x}_3
            \end{pmatrix}$. Thus, 
\begin{equation*}
    u^{*} = \frac{diag\left(m_1,m_2,m_3\right)}{\dot{x}_4 - \dot{z}_4}\left[g\begin{pmatrix}
        \dot{x}_4-\dot{x}_1 \\
                \dot{x}_4-\dot{x}_2 \\
                \dot{x}_4-\dot{x}_3
    \end{pmatrix}-\begin{pmatrix}
        \dot{x}_4\dot{z}_1 - \dot{x}_1\dot{z}_4 \\
        \dot{x}_4\dot{z}_2 - \dot{x}_2\dot{z}_4 \\
        \dot{x}_4\dot{z}_3 - \dot{x}_3\dot{z}_4
    \end{pmatrix}\right].
\end{equation*}

We have simulated the closed-loop control system with the preferred feedback control law using a standard fourth-order Runge-Kutta method where particle 4 makes an unconstrained motion while the other 3 are controlled for an alignment in their velocities. The closed-loop system evolves into the $xz$ plane. The initial positions of the particles are $(x_1,z_1)=(10,56),$ $(x_2,z_2)=(30,100),$ $(x_3,z_3)=(50,100)$ and $(x_4,z_4)=(10,90)$ and initial velocities $(\dot{x}_1,\dot{z}_1)=(0.5,1),$ $(\dot{x}_2,\dot{z}_2)=(1,1),$ $(\dot{x}_3,\dot{z}_3)=(-1,-1)$ and $(\dot{x}_4,\dot{z}_4)=(0.6,0)$ in thousand units. 

In Figure \ref{trajectories} we plot the trajectories of the 4 particles. In Figure \ref{constraints}, we plot a 3d trajectory that represent the values of the constraint function $\Phi=[\phi^1,\phi^2,\phi^3]^T$. We have considered the all particles' masses $m_i=2$ and $g=10$. The total simulation time was $500$ seconds, with a time step of $0.01$ seconds, resulting in $50000$ steps. We would like to remark that the trajectories cross the line $z=0$, which under reasonable models, could correspond to the ground. We have opted to simulate motion past this point for visualization purposes to avoid starting with higher initial value for $z$, but we remark that the trajectories' profile would not have changed. Also, after visualizing Figure \ref{trajectories} one observes intersections in the $(x,z)$ plane between different trajectories, and one might think this should be inconsistent with velocity alignment. However, this is not the case. We stress that these particles do not collide and pass through the intersection point at very different times.

\begin{figure}[htb!]
        \centering
        \includegraphics[width=0.9\linewidth]{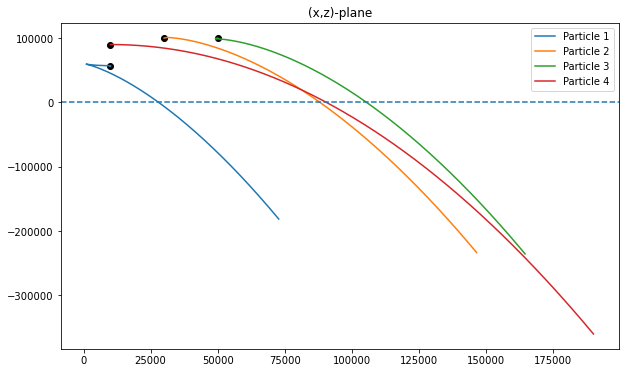}
        \caption{Trajectories of the 4 agents of the flocking closed-loop system. The black points indicate the initial positions.}
        \label{trajectories}
    \end{figure}

\begin{figure}[htb!]
        \centering
        \includegraphics[width=0.65\linewidth]{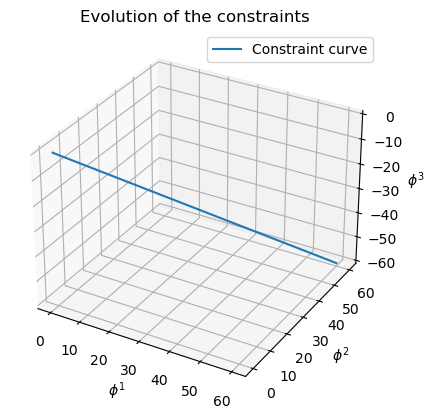}
        \caption{Constraint functions along the same trajectory.}
        \label{constraints}
    \end{figure}
%set_xticklabels: to change the labels of the graph

\subsection{Application to the control of a USV navigating a stream}

%\textcolor{magenta}{Define USV here?}

Consider an Unmanned Surface Vehicle (USV) consisting of a boat with a payload on the sea with a position-dependent stream. The position of the boat's center of mass is modeled by the configuration manifold $\mathbb{R}^{2}$ to which we add an orientation to obtain a complete description of its location in space, so that the system total configuration manifold is $\mathbb{R}^{2}\times \mathbb{S}$ with local coordinates $q=(x,y,\theta)$. The sea's current is modeled by the vector field $C:\mathbb{R}^{2}\rightarrow \mathbb{R}^{2}$, $C=(C^{1}(x,y), C^{2}(x,y))$.

The boat is well-modeled by a forced mechanical system with Lagrangian function
$L= \frac{m}{2}(\dot{x}^{2} + \dot{y}^{2}) + \frac{I}{2}\dot{\theta}^{2},$
where $m$ is the boat's mass, $I$ is the moment of inertia, and the external force is denoted by
$F^{ext} = W^{1} dx + W^{2} dy$
accounting for the action of the current on the center of mass of the boat and to which we add a control force
$F=u(\sin \theta dx - \cos \theta dy + d\theta ).$

The functions $W^{1}$ and $W^{2}$ are defined according to
$$\begin{cases}
    W^{1} &=  m  \ d \left(\sin^{2}\theta C^{1} - \sin\theta\cos\theta C^{2}\right)(\dot{q}) \\
    W^{2} &= m \ d \left(-\sin\theta\cos\theta C^{1} + \cos^{2}\theta C^{2} \right)(\dot{q}).
\end{cases},$$
where $d$ represents the differential of the functions inside the parenthesis. The external force assures that in the absence of controls, the dynamics of the boat satisfies the following kinematic equations
$$\begin{cases}
    \dot{x} = & \sin^{2}\theta C^{1} - \sin\theta\cos\theta C^{2} \\
    \dot{y} = & -\sin\theta\cos\theta C^{1} + \cos^{2}\theta C^{2},
\end{cases}$$
whenever the initial velocities in the $x$ and $y$ direction vanish. The corresponding controlled forced Lagrangian system is
\begin{equation*}
    m\ddot{x}=u \sin\theta + W^{1}, \quad m\ddot{y}=-u \cos\theta + W^{2}, \quad I\ddot{\theta}=u,
\end{equation*}
and, as we will show, it has the following virtual affine nonholonomic constraint
$$\sin\theta \dot{x} - \cos\theta \dot{y}= C^{2}\cos\theta - C^{1}\sin\theta.$$
The input distribution $\mathcal{F}$ is generated just by one vector field $Y=\frac{\sin \theta}{m}\frac{\partial}{\partial x}-\frac{\cos \theta}{m}\frac{\partial}{\partial y}+\frac{1}{I}\frac{\partial}{\partial \theta}$,
while the virtual nonholonomic constraint is the affine space $\mathcal{A}$ modelled on the distribution $\mathcal{D}$ defined as the set of tangent vectors $v_{q}\in T_{q}Q$  where $\mu(q)(v)=0,$ with $\mu=\sin\theta dx - cos\theta dy$. Thus, we may write it as
$\mathcal{D}=\hbox{span}\Big{\{} X_{1}=\cos \theta\frac{\partial}{\partial x} + \sin\theta \frac{\partial}{\partial y},\, X_{2}=\frac{\partial}{\partial \theta}\Big{ \}}$. The affine space is given as the zero set of the function $\phi(q,v) = \mu(q)(v) + Z(q)$ with
$Z(q) = \cos \theta C^{2}(x,y) - \sin \theta C^{1}(x,y)$
or, equivalently, as the set of vectors $v_{q}$ satisfying $v_{q}-C(q)\in \mathcal{D}_{q}$.

We may check that $\mathcal{A}$ is controlled invariant for the controlled Lagrangian system above. In fact, the control law
$\hat{u}(x,y,\theta,\dot{x},\dot{y},\dot{\theta})=-m\dot{\theta}(\cos\theta \dot{x} +\sin \theta \dot{y})$
makes the affine space invariant under the closed-loop system, since in this case, the dynamical vector field arising from the controlled Euler-Lagrange equations given by
\begin{align*}
    \Gamma = & \dot{x}\frac{\partial}{\partial x} + \dot{y}\frac{\partial}{\partial y} + \dot{\theta}\frac{\partial}{\partial \theta}+  \\ & \left(\frac{\hat{u}\sin \theta + W^{1}}{m}\right)\frac{\partial}{\partial \dot{x}} + \left( - \frac{\hat{u}\cos \theta - W^{2}}{m} \right)\frac{\partial}{\partial \dot{y}} + \frac{\hat{u}}{I}\frac{\partial}{\partial \dot{\theta}}
\end{align*}
is tangent to $\mathcal{A}$. This is deduced from the fact that $$\Gamma(\sin\theta \dot{x} - \cos\theta \dot{y} + \cos \theta C^{2}(x,y) - \sin \theta C^{1}(x,y))=0.$$

Let us find the control law that stabilizes the system using Theorem \ref{main:theorem}, namely $u_{a}^{*} = C_{a b}(-\phi^{b}-G(\phi^{b}))$. 
\begin{align*}
    C^{ab}&=(Y^a)^V(\phi^b) =\left(\frac{\sin \theta}{m}\frac{\partial}{\partial \dot{x}}-\frac{\cos \theta}{m}\frac{\partial}{\partial \dot{y}}+\frac{1}{I}\frac{\partial}{\partial \dot{\theta}}\right) \times \\ & \quad \quad \times\left(\sin\theta \dot{x} - \cos\theta \dot{y} - C^{2}\cos\theta + C^{1}\sin\theta\right)= \frac{1}{m}.
\end{align*}

The vector field $G$ which is determined by the unactuated forced mechanical system is given by
\[G=\dot{x}\frac{\partial}{\partial x} + \dot{y}\frac{\partial}{\partial y} + \dot{\theta}\frac{\partial}{\partial\theta} + W^1\frac{\partial}{\partial\dot{x}} + W^2\frac{\partial}{\partial\dot{y}} + 0\frac{\partial}{\partial\dot{\theta}},\]
thus, for $\phi=\sin\theta \dot{x} - \cos\theta \dot{y} - C^{2}\cos\theta + C^{1}\sin\theta$ we have 
\begin{align*}
    G(\phi)= &-\dot{x}\partial_xC^2c\theta + \dot{x}\partial_xC^1s\theta - \dot{y}\partial_yC^2c\theta + \dot{y}\partial_yC^1s\theta \\
    & + \dot{\theta}(\dot{x}c\theta +\dot{y}s\theta + C^2s\theta + C^1c\theta)+ W^1s\theta - W^2c\theta,
\end{align*}
where $\partial_xC^i=\frac{\partial C^i}{\partial x}$ and $\partial_yC^i=\frac{\partial C^i}{\partial y}$ for $i=1,2,$ $s\theta=\sin\theta$ and $c\theta=\cos\theta.$

Ultimately, $u_{a}^{*} = C_{a b}(-\phi^{b}-G(\phi^{b}))$ reads
\begin{align*}
    u^{*} = & -ms\theta\left[\dot{x} + C^1 + \dot{x}\partial_xC^1 + \dot{y}\partial_yC^1 + \dot{y}\dot{\theta} + \dot{\theta}C^2 + W^1 \right] \\
    & -mc\theta\left[-\dot{y} - C^2 - \dot{x}\partial_xC^2 - \dot{y}\partial_yC^2 + \dot{x}\dot{\theta} + \dot{\theta}C^1 - W^2 \right].
\end{align*}

Next, we present simulations of the closed-loop control system with the stability feedback control law using a standard fourth-order Runge-Kutta method. We have simulated two different cases. In the fist case, the boat is on a north-east current described by $C(x,y)=(1,1)$ measured in units per second. The initial position of the boat is $(x,y)=(1,1)$ and its orientation $\theta=\frac{\pi}{2}$ with initial velocity $\dot{q}=(\dot{x},\dot{y},\dot{\theta})=(0.8, 0.5,0)$. The mass of the boat is $m=10$ and the moment of inertia $I=1.5$. The simulation time was $100$ seconds, with a time step of $0.01$ seconds, resulting in $10000$ steps.

\begin{figure}[htb!]
        \centering
        \includegraphics[width=0.9\linewidth]{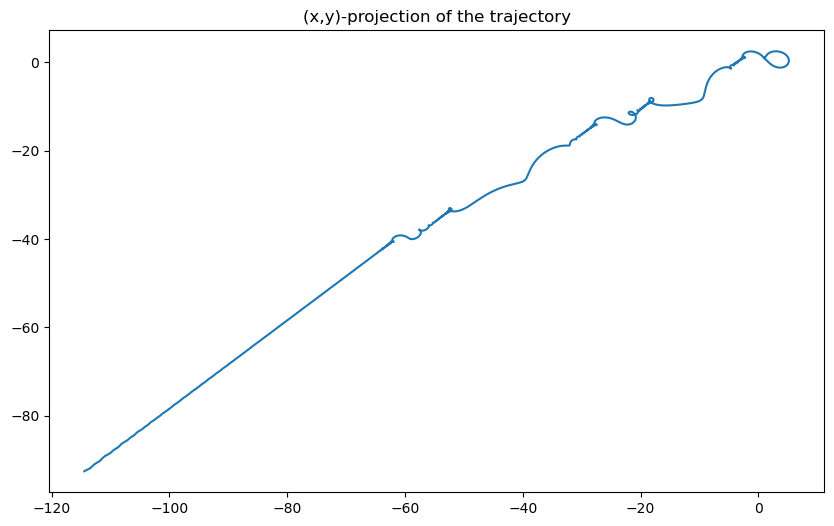} 
        \caption{Left: Projection of a trajectory of the closed-loop system into the plane $xy$ of the boat on a north-east stream. }
        \label{boat:traj:n-e}
\end{figure}

%\begin{figure}[htb!]
 %       \centering
        %\includegraphics[width=0.4\linewidth]{boatconstraintne.png}
        %\caption{Constraint $\phi(t)$.}
        %\label{boat:const:n-e}
%\end{figure}

\begin{figure}[htb!]
        \centering
        \includegraphics[width=0.49\linewidth]{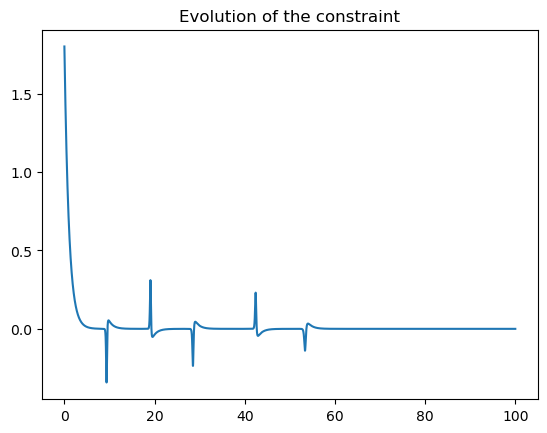}\includegraphics[width=0.49\linewidth]{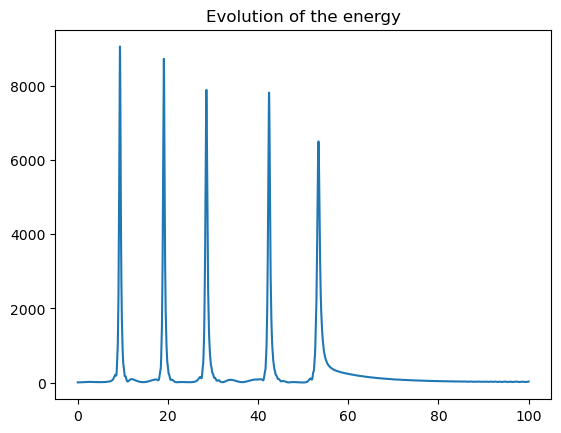}
        \caption{Left: Constraint $\phi(t)$. Right: Energy of the system.}
        \label{boat:const-energy:n-e}
\end{figure}

In Figure \ref{boat:traj:n-e} we graph the trajectory of the boat under the effect of the current and the control force. In Figure \ref{boat:const-energy:n-e} we plot the evolution of the constraint and the energy of the system over time. Note that a significant amount of energy is used to force the boat comply with the constraints.

In the second case, the current of the sea is an anticyclone (high-pressure area) described by $C(x,y)=(y,-x+y)$. The initial position and orientation of the boat was the same as before, $q=(x,y,\theta)=(1,1,\frac{\pi}{2})$, with initial velocity $\dot{q}=(\dot{x},\dot{y},\dot{\theta})=(1, 1,0)$. The mass of the boat is $m=20$ and the moment of inertia $I=4$. The simulation time was set to $100$ seconds, with a time step of $0.01$ seconds, resulting in $10000$ steps. In Figure \ref{boat:traj:antic} we plot the trajectory of the boat on the $xy$ plane under the effect of the current and the control force. In Figure \ref{boat:const:antic} we plot the evolution of the constraint and the energy of the system over time.

\begin{figure}[htb!]
        \centering
        \includegraphics[width=0.9\linewidth]{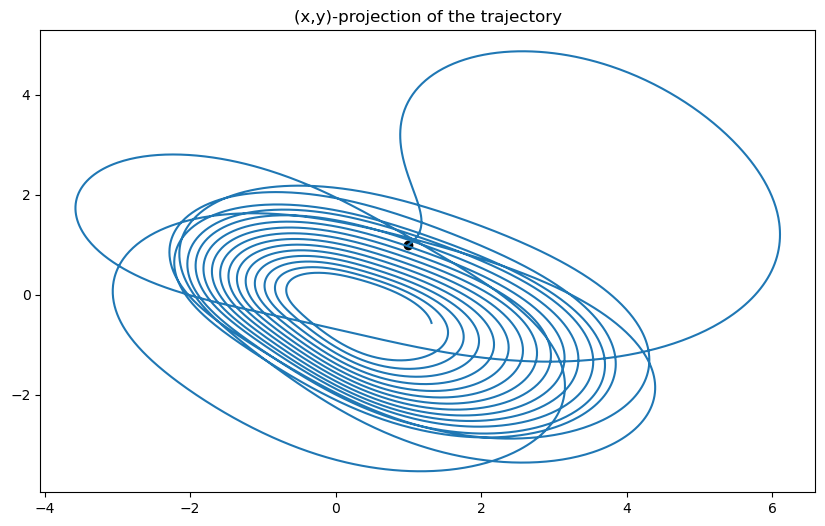 } 
        \caption{Projection of a trajectory of the closed-loop system into the plane $xy$ of the boat on an anticyclone stream. The black dot indicates the initial position of the boat. }
        \label{boat:traj:antic}
\end{figure}

\begin{figure}[htb!]
        \centering
        \includegraphics[width=0.49\linewidth]{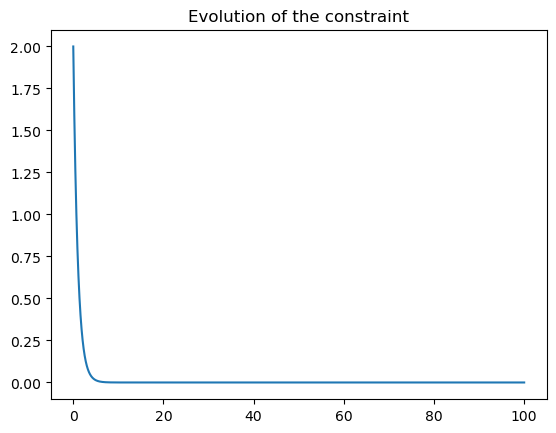} \includegraphics[width=0.49\linewidth]{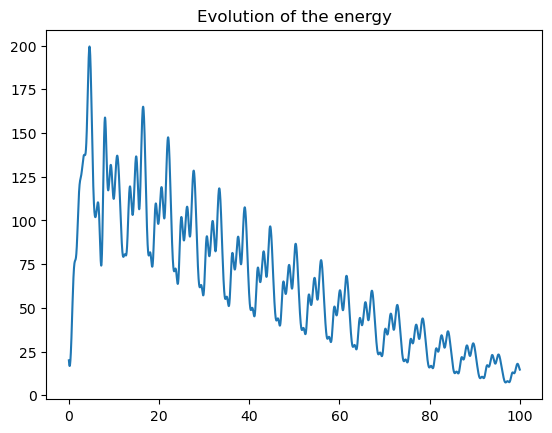}
        \caption{Left: Constraint $\phi(t)$. Right: Energy of the system.}
        \label{boat:const:antic}
\end{figure}

\color{black}

\section{Conclusions and Future Work}\label{conc}
We investigated sufficient conditions and developed a control law driving a mechanical control system to converge exponentially to a desired submanifold of the phase space determined by virtual nonholonomic constraint. Two examples were presented to validate the theoretical results.
 
 Some related questions to the main result of the paper have to be addressed in future work. This first is the development of a control law that stabilizes the constraint distribution for the case where a mechanical system evolves in a Lie group \cite{liegroups} or a homogeneous space \cite{homogeneous}. An interesting question is to understand the qualitative behavior of the closed-loop dynamics with a special emphasis on the energy. We have observed numerically that the energy often stabilizes around a specific value or it converges to a bounded set of values. Understanding the limit value of the energy might give a clue on the nature of the closed-loop dynamics. %In addition, we would like to characterize it geometrically as in \cite{VNNHC}.

\end{document}